\documentclass[conference]{IEEEtran}
\IEEEoverridecommandlockouts

\usepackage{cite}
\usepackage{amsmath,amssymb,amsfonts}
\usepackage{algorithmic}
\usepackage{graphicx}
\usepackage{textcomp}
\usepackage{xcolor}
\usepackage{amsthm}
\def\BibTeX{{\rm B\kern-.05em{\sc i\kern-.025em b}\kern-.08em
    T\kern-.1667em\lower.7ex\hbox{E}\kern-.125emX}}

\allowdisplaybreaks
\newtheorem{theorem}{Theorem}
\newtheorem{lemma}{Lemma}
\newtheorem{corollary}{Corollary}

\definecolor{newred}{HTML}{ff382e}

\def\B{\mathbb{B}}

\def\F{\mathbb{F}}

\title{RPA Decoding with Few Subspace Projections for Correcting Structured Error Patterns in RM Codes}
\author{%
  \IEEEauthorblockN{Shiv Pratap Singh Rathore \, and \, Navin Kashyap}
  \IEEEauthorblockA{Department of Electrical Communication Engineering \\
                   Indian Institute of Science \\
                    Email: \{shivprataps,nkashyap\}@iisc.ac.in}
}
\date{}

\begin{document}

\maketitle

\begin{abstract}
The Recursive Projection Aggregation (RPA) decoder for Reed-Muller codes uses, in its default implementation, all one-dimensional subspaces at each level of recursion. We restrict our attention to the case where the number of errors is $t<d_{\mathrm{min}}/2$. It is known that if the number of errors is $t<d_{\mathrm{min}}/2$, then $2t+1$ subspaces at each level of recursion are sufficient to correct these errors. We show that, for structured error patterns, the number of subspaces required for correction can be reduced further by appropriately selecting the subspaces used for projection. In particular, when errors occur as a single burst of length $t < d_{\mathrm{min}}/2$, three subspaces at each level of recursion suffice for correction. We obtain similar results for multiple bursts and other structured error patterns as well.
\end{abstract}

\section{Introduction}
Reed-Muller (RM) codes are among the simplest and oldest families of error-correcting codes, introduced by Muller in 1954 \cite{muller1954application}. In recent years, RM codes have attracted renewed interest owing to theoretical progress demonstrating that they achieve capacity over the binary erasure channel (BEC)  and binary-input memoryless symmetric (BMS) channels \cite{kudekar2016reed,reeves2023reed,abbe2023proof}. More recently, they have also been shown to achieve the symmetric capacity of binary-input indecomposable finite-state channels \cite{pfister2026reedmullercodesachievesymmetric}.

Much work has been dedicated to developing practical and efficient decoding algorithms for RM codes. The first decoder was devised by Reed \cite{1057465} and was capable of correcting bit-flip errors up to half the minimum distance. For first-order RM codes, a Fast Hadamard Transform-based (FHT) decoder \cite{be1986optimal} was proposed as an efficient Maximum Likelihood (ML) decoder. Various other decoders have been proposed for RM codes \cite{sidel1992decoding,sakkour2005decoding,dumer2004recursive,dumer2006soft,sberlo2020performance,geiselhart2021automorphism,kamenev2019new,kamenev2021decoding,doan2022decoding}. For more information about RM codes and their decoders, we refer the reader to \cite{abbe2020reed} . In this paper, we focus on a particular decoder introduced by Ye and Abbe known as the Recursive Projection-Aggregation (RPA) decoder \cite{ye2020recursive}, which projects the received vector onto lower-order quotient codes, recursively decodes them, and aggregates the estimates via majority voting. The RPA decoder has been shown to achieve near-maximum likelihood (ML) performance over the binary symmetric channel (BSC) at moderate blocklengths.

The RPA decoder for the $\mathrm{RM}(m,r)$ code uses all $2^m-1$ one-dimensional subspaces for projection and has complexity $O(n^r \log n)$, where $n = 2^m$ is the blocklength of the RM code. In this paper, we focus on the case where the number of errors is less than half the minimum distance. For a code $\mathrm{RM}(m,r)$ with minimum distance $d_{\mathrm{min}} = 2^{m-r}$, it is known that if the number of errors $t$ satisfies $t < \frac{d_{\mathrm{min}}}{2}$, then any arbitrary error pattern of weight $t$ can be corrected using $2t+1$ one-dimensional subspaces within the RPA framework \cite{zhang2025coset}. Using $2t+1$ subspaces instead of all $2^{m}-1$ subspaces reduces the complexity to $O((2t+1)^{r-1} n \log n)$. However, when $t$ is large, the computational complexity can still be substantial.

In many practical communication scenarios, channel errors are not arbitrary but possess specific structures, such as localized burst errors. In this paper, we focus on correcting structured error patterns using fewer subspaces. We show that if the error positions satisfy certain constraints, then even $2t+1$ subspaces are unnecessary, and $t$ errors can be corrected using a constant number of subspaces for projection. Burst errors are a special case of the structured error patterns considered in this paper, and we show that when errors occur as $\beta$ bursts with a total of $t$ errors where $t < \frac{d_{\mathrm{min}}}{2}$, then we can correct them using $2\beta+1$ subspaces, provided that $\beta$ satisfies an upper bound determined by the burst lengths. For example, if errors occur as a single burst of length $t < \frac{d_{\mathrm{min}}}{2}$ , then we can correct them using three subspaces, with complexity $O(3^r \cdot n)$.

\section{Notation and Preliminaries}

\subsection{Notations}
Let $\F_2$ and $\F_2^m$ be the finite field of two elements and the $m$-dimensional vector space over $\F_2$, respectively, where $m$ is a positive integer. We denote modulo-$2$ addition by $\oplus$. A vector $z \in \F_2^m$ is represented as $z=(z_1,\ldots,z_m)$ where $z_j \in \{0,1\}$ for $j=1,\ldots,m$. For two vectors $a=(a_1,\ldots,a_m), \ b=(b_1,\ldots,b_m) \in \F_2^m$, we have $a=b$ iff $a_j=b_j \ \forall j \in \{1,\ldots,m\}$. Also, we define $a \oplus b=(a_j \oplus b_j :  j \in \{1,\ldots,m\})$.

We order the elements of $\mathbb{F}_2^m$ lexicographically. In particular, for $z=(z_1,\ldots,z_m)$ and $z'=(z'_1,\ldots,z'_m)$ in $\mathbb{F}_2^m$, we say that $z$ precedes $z'$ if, at the first coordinate $j$ for which $z_j\neq z'_j$, we have $z_j=0$ and $z'_j=1$. 

For a vector $z=(z_1,\ldots,z_m) \in \F_2^m$, we use $z \langle i \rangle $ to denote the last $i$ coordinates of $z$, i.e., $z\langle i \rangle=(z_{m-i+1},z_{m-i+2},\ldots,z_m)$. For two vectors $z \in \F_2^m, z' \in \F_2^{m'}$, where $m$ need not be equal to $m'$, we say $z\langle i \rangle=z'\langle i \rangle$ if $(z_{m-i+1},z_{m-i+2},\ldots,z_{m})=(z'_{m'-i+1},z'_{m'-i+2},\ldots,z'_{m'})$. 


\subsection{Reed-Muller Codes}                                     
Let $ \F_2[x_1, \ldots, x_m] $ be the ring of polynomials in $ m $ variables over $ \F_2 $. For any polynomial $ f $ and any binary vector $ z = (z_1, \ldots, z_m) \in \F_2^m $, the evaluation of $f$ at $z$ is defined as $ \text{Eval}_z(f) := f(z_1, \ldots, z_m)$.
The evaluation vector $ \text{Eval}(f) $ consists of all the values obtained by evaluating $f$ at every binary vector in $\F_2^m$:
$$
\text{Eval}(f) = \bigl(\text{Eval}_z(f): z \in \F_2^m \bigr).
$$
where the coordinates $z$ are ordered according to the standard lexicographic order.

For integers $0 \leq r \leq m$, the Reed–Muller code of $m$ variables and order $r$, denoted by $ \mathrm{RM}(m, r)$, is defined as the set consisting of the evaluation vectors of all the polynomials in $\F_2[x_1, \ldots, x_m]$ of degree at most $r$:
\begin{equation*}
\mathrm{RM}(m, r) := \left\{ \text{Eval}(f) : \, f \in \mathbb{F}_2[x_1, \ldots, x_m],\ \deg(f) \leq r \right\}.
\end{equation*}
The code $\mathrm{RM}(m,r)$ is of length $2^m$, dimension $\sum\limits_{i=0}^{r} {m \choose i}$ and minimum distance $d_{\mathrm{min}}=2^{m-r}$.

\subsection{Projection and Quotient Code}
Let $\B$ be an $s$-dimensional subspace of $\F_2^m$. The quotient space $\F_2^m / \B$ consists of all the cosets of $\B$ in $\F_2^m$, each coset $T$ being expressible as $ T = z \oplus \B := \{z \oplus b: b \in \B\} $ for some $z \in \F_2^m$.
For a binary vector $v = (v_z: z \in \F_2^m)$, the projection of $v$ onto the cosets of $\B$ is defined as
\begin{equation*}
v_{/\B} := \text{Proj}(v,\B) := (v_{/\B}(T): T \in \F_2^m/\B),
\label{eq:projv}
\end{equation*}
where $v_{/\B}(T) = \bigoplus_{z \in T} v_z$. 

 The coordinate ordering in the projection vector $v_{/\B}$ is prescribed as follows. Fix an $(m-s) \times m$ parity-check matrix, $H$, for $\B$. As is standard, the \emph{syndrome} of a vector $z \in \F_2^m$ is defined as $H z^{\mathsf{T}}$. Vectors belonging to the same coset of $\B$ have the same syndrome. We can then order the cosets of $\B$ according to the lexicographic ordering of their syndromes, and this is the coordinate ordering followed in $v_{/\B}$. For a coset $T$, we use $[T]$ to denote its index in the projection vector.

 When $\B = \{0,b\}$ is a one-dimensional subspace, we use a specific choice of parity-check matrix $H$. Let $j$ be the largest index such that $b_j = 1$, i.e., $b = (b_1,\ldots,b_{j-1},1,0,\ldots,0)$. We then take the $(m-1) \times m$ parity-check matrix to be
$$
H = \begin{bmatrix}
I_{(j-1) \times (j-1)} & \widetilde{b} & \mathbf{O} \\
\mathbf{O} & \mathbf{0} & I_{(m-j) \times (m-j)}
\end{bmatrix}
$$
where $\mathbf{0}$ is a column vector of $(m-j)$ $0$s, $\widetilde{b} = [b_1,\ldots,b_{j-1}]^{\mathsf{T}}$, and the $\mathbf{O}$s denote all-zero matrices of appropriate sizes. Note that $Hz^{\mathsf{T}} = [\widetilde{z_1},\ldots,\widetilde{z}_{j-1},z_{j+1},\ldots,z_{m}]^{\mathsf{T}}$, where $\widetilde{z}_\ell = z_\ell \oplus z_j b_\ell$ for $\ell = 1,\ldots,j-1$. In particular, $H b^{\mathsf{T}} = [0,\ldots,0]^{\mathsf{T}}$.

For example, let $v=(v_{000},v_{001},\ldots,v_{111})$ be a binary vector indexed by $\F_2^3$ and consider the one-dimensional subspace $\B=\{000,110\}$. Then, $Hz^{\mathsf{T}}= [z_1 \oplus z_2, z_3]^{\mathsf{T}}$. Consider coset $111 \oplus \B=\{001,111\}$. For $z=111$, $Hz^{\mathsf{T}}=[0,1]^{\mathsf{T}}$, thus the index of coset $111 \oplus \B$ in $v_{/\B}$ is given by $[111 \oplus \B]=01$. The projection vector is $v_{/\B}=(v_{000} \oplus v_{110},v_{001} \oplus v_{111}, v_{010} \oplus v_{100}, v_{011} \oplus v_{101})$. 

The \emph{quotient code} of $\mathrm{RM}(m,r)$ on the $s$-dimensional subspace $\B$, denoted by $Q(m,r,\B)$, is the set consisting of projections of all the codewords in $\mathrm{RM}(m,r)$, i.e., $Q(m,r,\B)=\{v_{/\B}: v \in \mathrm{RM}(m,r)\}$. It is known that $Q(m,r,\B)$ is simply $\mathrm{RM}(m-s,r-s)$ \cite{abbe2020reed}. Thus, for $v \in \text{RM}(m,r)$, $v_{/\B} \in \text{RM}(m-s,r-s)$. The RPA decoder is based on this fact \cite{ye2020recursive}.

\subsection{Recursive Projection Aggregation (RPA) decoder}
We consider the RPA decoder over a binary symmetric channel (BSC). Consider transmitting $c \in \mathrm{RM}(m,r)$. Let $y=(y(z), \ z \in \F_2^m)$ be the received vector. Let $\mathcal{S}$ be a subset of all $s$-dimensional subspaces of $\F_2^m$. The RPA decoder works iteratively with each iteration consisting of three steps: 

\begin{enumerate}
    \item \textbf{Projection:} For every $s$-dimensional subspace $\B \in \mathcal{S}$, calculate the projection $y_{/\B}$. The projected vector $y_{/\B}$ is a corrupted version of the quotient codeword $c_{/\B}\in Q(m,r,\B)$.

    \item \textbf{Decoding of quotient codes:} Decode $y_{/\B}$ for every $\B$ recursively using RPA decoder for the corresponding quotient code $Q(m,r,\B)$. Let $\hat{y}_{/\B}$ denote the resulting estimate of $c_{/\B}$.

    \item \textbf{Aggregation/Voting:} For each $z\in\F_2^m$, define $\widehat{y}(z,\B)=\hat{y}_{/\B}([z \oplus \B])\oplus y_{/\B}([z\oplus \B])\oplus y(z)$, where $[z\oplus \B]$ denotes the index of the coset $z\oplus\B$ in the projection vector. The estimates from the different subspaces are then combined by majority voting:
   $$\hat{y}(z)=\boldsymbol{1}\left\{\sum_\B \widehat{y}(z,\B)>\frac{h}{2}\right\},\qquad z\in\mathbb{F}_2^m.$$
   Here, $h$ denotes $|\mathcal{S}|$, i.e., the number of $s$-dimensional subspaces used for decoding.

\end{enumerate}

We only consider one-dimensional RPA (RPA-1D), where the received vector is projected onto one-dimensional subspaces in $\mathcal{S}$, with $\mathcal{S}$ being a subset of all $2^m-1$ one-dimensional subspaces of $\F_2^m$. Then, the quotient code $Q(m,r,\B)$ is $\mathrm{RM}(m-1,r-1)$. For $r=2$, the quotient code is a first-order Reed-Muller code and is decoded using the FHT decoder, which serves as the base case of the recursion. For $r>2$, the quotient code is decoded recursively using RPA-1D. Thus, at each recursive level, both the number of variables and the order are reduced by one until the first-order RM code is reached, where the FHT decoder is used. 

After one iteration, if $\hat{y}=y$, the decoder terminates and outputs $\hat{c}=\hat{y}$. Otherwise, $y$ is replaced by $\hat{y}$, and the next iteration is performed. The number of iterations is limited to $N_{\max}$ to avoid an infinite loop. In this paper, we fix $N_{\max}=1$. In this case, after a single iteration, the decoder outputs $\hat{c}=\hat{y}$.

\section{Subspace Selection}
Here we consider the case where the number of errors is $t$ and $t < \frac{d_{\mathrm{min}}}{2}$. It is known that if the number of errors is $t < \frac{d_{\mathrm{min}}}{2}$, then such errors can be corrected by the RPA decoder using $2t+1$ one-dimensional subspaces \cite{zhang2025coset}. We show that if errors follow specific error patterns, then we can correct them using fewer subspaces. 

We first prove a lemma that will be used later.

\begin{lemma} \label{lem:projvector}
    Consider a binary vector $v=(v(z): z \in \F_2^m)$ and a one-dimensional subspace $\B=\{0,b\}$, where the last $i$ coordinates of $b$ are zero, i.e., $b=(b_1,\ldots,b_{m-i},0,\ldots,0)\in \F_2^m$ for some integer $i$ satisfying $0 \leq i \leq m-1$. Then,
    \begin{enumerate}
        \item[(i)] For every coset $T=z \oplus \B=\{z,z \oplus b\}$, the last $i$ coordinates of $z$ and $z \oplus b$ are the same, i.e., $z\langle i \rangle=(z \oplus b)\langle i \rangle$.
        \item[(ii)] If $[z\oplus\B] \in \F_2^{m-1}$ denotes the index of the coset $z\oplus\B$ in the projection vector $v_{/\B}$, then the last $i$ coordinates of $z$ and $[z \oplus \B]$ are the same, i.e., $z\langle i \rangle=[z \oplus \B]\langle i \rangle$
    \end{enumerate}
\end{lemma}

\begin{proof}

    (i) Since the last $i$ coordinates of $b$ are all zeros, i.e., $b_{m-i+1}=b_{m-i+2}\ldots=b_{m}=0.$ Thus, for any $z \in \F_2^m$, $(z \oplus b)_{j}=z_{j} \oplus b_{j} = z_{j} \ \forall j \in \{m-i+1,m-i+2,\ldots,m\}$. Thus, the last $i$ coordinates of $z$ and $z\oplus b$ are the same. \\

    (ii) For the indexing of cosets in the projection vector $v_{/\B}$, the cosets are ordered according to the lexicographic ordering of their syndromes. We have seen that for a $z \in \F_2^m$, its syndrome is given as $Hz^{\mathsf{T}} = [\widetilde{z_1},\ldots,\widetilde{z}_{j-1},z_{j+1},\ldots,z_{m}]^{\mathsf{T}}$, where $j$ is the largest index such that $b_j=1$ and $\widetilde{z}_\ell = z_\ell \oplus z_j b_\ell$ for $\ell = 1,\ldots,j-1$. Then, the last $m-j$ coordinates of $z$ and the syndrome $Hz^{\mathsf{T}}$ are the same, i.e., $z \langle m-j \rangle=[z \oplus \B] \langle m-j \rangle$. Since the last $i$ coordinates of $b$ are all zeros, we have $j \leq m-i$. Therefore, the last $i$ coordinates of the syndrome $Hz^{\mathsf{T}}$ are the same as the last $i$ coordinates of $z$, i.e., $z\langle i \rangle=[z \oplus \B]\langle i \rangle$.
\end{proof}

 For example, let $v=(v_{0000},v_{0001},\ldots,v_{1111})$ be a binary vector indexed by $\F_2^4$, consider the one-dimensional subspace $\B=\{0000,0100\}$. Now, $Hz^{\mathsf{T}}=[z_1,z_3,z_4]^{\mathsf{T}}$. Thus, $v_{/\B}=(v_{0000} \oplus v_{0100},v_{0001} \oplus v_{0101}, v_{0010} \oplus v_{0110}, v_{0011} \oplus v_{0111}, v_{1000} \oplus v_{1100},v_{1001} \oplus v_{1101}, v_{1010} \oplus v_{1110}, v_{1011} \oplus v_{1111})$. It can be seen that Lemma~\ref{lem:projvector} holds. For example, consider $z=1110 \in \F_2^4$, then $1110 \oplus b=1010$. Then, $z\langle 2 \rangle=(z \oplus b)\langle 2 \rangle$. Also, $Hz^{\mathsf{T}}=[1,1,0]^{\mathsf{T}}$, i.e., $[z \oplus \B]=110$. Thus, $z\langle 2 \rangle=[z \oplus \B]\langle 2 \rangle$.  \\

Let $E=\{e^1,e^2,\ldots, e^t\}$ denote the set of error positions in the received vector, where $e^j \in \F_2^m$ and $e^1,e^2,\ldots,e^t$ are in lexicographic order.

For a fixed $\ell \in \{1,2,\ldots,m-r\}$ and $\gamma \in \F_2^{m-r-\ell}$, we define $$\pi(\gamma)= \#\{e^{j}: e^j \langle m-r-\ell \rangle= \gamma \}.$$
We consider the error patterns where $\pi(\gamma) \leq \beta \ \forall \gamma \in \F_2^{m-r-\ell},$ for some integer $\beta$. We will show that these errors can be corrected by the RPA decoder using $2\beta+1$ one-dimensional subspaces for a suitable choice of subspaces, provided $\beta$ satisfies an upper bound determined by $\ell$. 

Now, let us define the notation that we will be using for the different levels of recursion in the RPA decoder. Let $c \in \mathrm{RM}(m,r)$ be the transmitted codeword and $y$ be the received vector.  We refer to $c$ and $y$ as the 0-th level projections, denoted by $c^{(0)}=c$ and $y^{(0)}=y$. At each recursion level, the one-dimensional subspaces used for projection are chosen from a set of candidate subspaces satisfying the required properties. To simplify the notation, without loss of generality, we fix a path or sequence of one-dimensional subspaces, with one subspace corresponding to each recursion level. The $i$-th level is obtained by recursively projecting the code $\mathrm{RM}(m,r)$ using the first $i$ one-dimensional subspaces along this path, resulting in the quotient code $\mathrm{RM}(m-i,r-i)$. The specific path is not important since our results hold for any choice of subspaces from the set of candidate subspaces. Hence, we omit the path from the notation and let $c^{(i)}$ and $y^{(i)}$ denote the projections of $c$ and $y$, respectively, obtained after $i$ recursive projections along the fixed path. Note that $c^{(i)} \in \mathrm{RM}(m-i,r-i)$. Let $E^{(i)}$ denote the set of error positions at the $i$-th level, i.e., the set of positions $z \in \F_2^{m-i}$ for which $y^{(i)}(z) \neq c^{(i)}(z)$. We denote the number of errors at this level by $t_i=|E^{(i)}|$. For $\gamma \in \F_2^{m-r-\ell}$, we define 
$$ \pi^{(i)}(\gamma) = \# \{e \in E^{(i)} : e \langle m-r-\ell \rangle = \gamma \} $$


In the traditional RPA decoder, $r=1$ is used as the base case, which means that the recursion stops when the quotient code obtained after the projection step is a first-order RM code. The FHT decoder, which has a complexity of $O(n \log n)$, is used to decode the first-order RM code. In this paper, we consider two variants. The first uses $r=0$ as the base case, while the second uses the standard $r=1$ as the base case.
 Consider the variant where we use $r=0$ as the base case. Since the quotient code at this level is a repetition code, we use majority voting. In this case, for decoding first-order RM codes, we project down to $r=0$ using $2\beta+1$ subspaces. The complexity for each of the projection, majority decoding at $r=0$ and aggregation steps is $O \big((2\beta+1) \cdot n \big)$. If $\beta$ is a constant, we can achieve $O(n)$ complexity for decoding the first-order RM code.
 
For a code $\mathrm{RM}(m,r)$, we use one-dimensional subspaces $\B=\{0,b\}$ where $b \in \F_2^m$, $b \ne 0$, is chosen such that its last $m-r-\ell$ coordinates are zero. There are $2^{r+\ell}-1$ nonzero vectors with this property, and each one-dimensional subspace contains exactly one such nonzero vector. Thus, the total number of such valid subspaces is $2^{r+\ell}-1$. At the $i$-th level of projection, the quotient code is $\mathrm{RM}(m-i,r-i)$. Then, the number of nonzero vectors in $\F_2^{m-i}$ with their last $m-r-\ell$ coordinates being zero is $2^{r+\ell-i}-1$. 

If we use $r=0$ as the base case, then the last recursive step projects the code $\mathrm{RM}(m-r+1,1)$ onto $\mathrm{RM}(m-r,0)$. Thus, at least $2^{\ell+1}-1$ valid subspaces are available for projection at any given level. This allows us to select $2 \beta+1 \leq 2^{\ell+1}-1$, which implies we require $\beta \leq 2^{\ell}-1$. If we instead use $r=1$ as the base case, then the last recursive step projects the code $\mathrm{RM}(m-r+2,2)$ onto $\mathrm{RM}(m-r+1,1)$. In this case, at least $2^{\ell+2}-1$ valid subspaces are available for projection at any level, allowing us to select $2 \beta+1 \leq 2^{\ell+2}-1$, meaning $\beta \leq 2^{\ell+1}-1$. Therefore, using $r=1$ as the base case allows our framework to handle a larger value of $\beta$, but at the cost of higher decoding complexity. 

Let us first consider the RPA decoder that uses $r=0$ as the base case.

\begin{theorem} \label{thm:main}
    For a code $\mathrm{RM}(m,r)$, suppose the number of errors is $t < \frac{d_{\mathrm{min}}}{2}$. Let $\pi(\gamma)= \#\{e \in E: e \langle m-r-\ell \rangle= \gamma \}$ for some fixed $\ell \in \{1,2,\ldots,m-r\}$ and $\gamma \in \F_2^{m-r-\ell}$. If $\pi(\gamma) \leq \beta$ holds for all $\gamma\in\F_2^{m-r-\ell}$, where $\beta \le 2^{\ell}-1$ is some fixed integer, then these errors can be corrected by the RPA decoder using $2\beta+1$ one-dimensional subspaces at each level of recursion, with the repetition code ($r=0$) as the base case.
\end{theorem}

\begin{proof}
We assume that at all levels of projection, the one-dimensional subspaces $\B=\{0,b\}$ used have the property that the last $(m-r-\ell)$ coordinates of $b$ are zero. Note that at the $i$-th level of projection, the corresponding quotient code is $\mathrm{RM}(m-i,r-i)$ and its minimum distance is $2^{(m-i)-(r-i)}=2^{m-r}$. Thus, at any level of projection, the minimum distance of the quotient code remains $d_{\mathrm{min}}=2^{m-r}$.

To prove the theorem, we establish the following two properties using induction: 
\begin{enumerate}
     \item For the $i$-th level of projection and any valid sequence or path of one-dimensional subspaces, the number of errors $t_i$ is strictly less than $\frac{d_{\mathrm{min}}}{2}$, and $\pi^{(i)}(\gamma) \leq \beta$ for all $\gamma \in \F_2^{m-r-\ell}$. This holds for all $i \in \{0,1,\ldots,r\}$.

    \item All projections at each level $0,1,\ldots,r$ are decoded correctly, i.e., for any valid sequence of one-dimensional subspaces, $\widehat{y}^{(i)}=c^{(i)} \ \forall i \in \{0,1,\ldots,r\}$.
\end{enumerate}

First, we prove Property~1 using forward induction. At the $0$-th level of projection, we are given $t_0=t<\frac{d_{\mathrm{min}}}{2}$ and $\pi^{(0)}(\gamma)=\pi(\gamma)\leq \beta$ for all $\gamma$. Hence, Property~1 holds at the $0$-th level. 

Without loss of generality, we fix an arbitrary path of one-dimensional subspaces used up to the $i$-th level. By the induction hypothesis, assume Property~1 holds at the $i$-th level. We show it holds at the $(i+1)$-th level. Let the number of errors at level $i$ and $(i+1)$ be $t_i$ and $t_{i+1}$, respectively. By the induction hypothesis, $t_i < \frac{d_{\mathrm{min}}}{2}$. Let $\B=\{0,b\}$ be the one-dimensional subspace used for projection from level $i$ to $i+1$, i.e., $y^{(i+1)}=(y^{(i)})_{/\B}$. Since each error position at level $i$ belongs to exactly one coset of $\B$, each error can contribute to at most one coordinate of the projected vector $y^{(i+1)}$. Therefore, $t_{i+1}\leq t_i < \frac{d_{\mathrm{min}}}{2}$. Furthermore, the set of error positions $E^{(i+1)}$ is a subset of $\{[e \oplus \B]: e \in E^{(i)}\}$. Since the last $m-r-\ell$ coordinates of $b$ are zero, Lemma~\ref{lem:projvector}(ii) implies $e \langle m-r-\ell \rangle =[e \oplus \B]\langle m-r-\ell \rangle \ \forall e \in E^{(i)}$. Therefore, $\pi^{(i+1)}(\gamma) \leq \pi^{(i)}(\gamma)$. By the induction hypothesis, $\pi^{(i)}(\gamma) \leq \beta$, and consequently $\pi^{(i+1)}(\gamma) \leq \beta.$ This proves Property~1.

Next, we prove Property~2 using backward induction on the levels from $r$ down to $0$. The base case is the $r$-th level of projection, where the quotient code is the repetition code $\mathrm{RM}(m-r,0)$. From Property~1, the number of errors $t_r$ is less than $d_{\mathrm{min}}/2=2^{m-r-1}$. Since $\mathrm{RM}(m-r,0)$ is a repetition code with blocklength $2^{m-r}$ and the number of errors is strictly less than half the blocklength, all errors can be corrected using majority voting. Thus, $\widehat{y}^{(r)}=c^{(r)}$.

By the induction hypothesis, assume that all child projections from level $i,i+1,\ldots,r$ are decoded correctly. We show that the projection at the $(i-1)$-th level is also decoded correctly. Let $\B^1,\ldots,\B^{2\beta+1}$ be the set of $2\beta+1$ subspaces used for the projection from level $i-1$ to level $i$, where $\B^j=\{0,b^j\}$.

Consider a non-erroneous position $z$ in $y^{(i-1)}$, i.e., $y^{(i-1)}(z)=c^{(i-1)}(z)$. For subspace $\B^j=\{0,b^j \}$, coset $z \oplus \B^j$ contains $z$ and $z \oplus b^j$. If $z \oplus b^j \not \in E^{(i-1)}$, then the projected coordinate $y^{(i)}([z \oplus \B^j])$ is not in error. By the induction hypothesis, it is decoded correctly at level $i$, so $\widehat{y}^{(i)}[z \oplus \B^j] = y^{(i)}([z \oplus \B^j])$. Thus, $\widehat{y}^{(i-1)}(z,\B^j)=y^{(i-1)}(z)=c^{(i-1)}(z)$, i.e., a correct estimate is obtained from subspace $\B^j$. 

Now consider an erroneous position $z \in E^{(i-1)}$. If $z \oplus b^j \not \in E^{(i-1)}$, then the projected coordinate $y^{(i)}([z \oplus \B^j])$ is in error. By the induction hypothesis, this projection at level $i$ is decoded correctly, thus $\widehat{y}^{(i)}[z \oplus \B^j] \neq y^{(i)}([z \oplus \B^j])$. Therefore,  $\widehat{y}^{(i-1)}(z,\B^j)=y^{(i-1)}(z) \oplus 1=c^{(i-1)}(z)$ , i.e., a correct estimate is obtained from subspace $\B^j$.

Thus, whether a position $z$ is erroneous or not, as long as it forms a coset with a non-erroneous position, we obtain a correct estimate for $\widehat{y}^{(i-1)}(z)$ from subspace $\B^j$. If $h=2\beta+1$ is the total number of subspaces used, then the final estimate is determined by majority voting: $\widehat{y}^{(i-1)}(z)=\mathbf{1}\left\{\sum \widehat{y}^{(i-1)}(z,\B^j)>\frac{h}{2}\right\}$. Therefore, $\widehat{y}^{(i-1)}(z)=c^{(i-1)}(z)$ if no more than $\lfloor h/2 \rfloor = \beta$ of the cosets $z \oplus \B^j$ have the position $z \oplus b^j$ being in error. Indeed, at most $\beta$ of the positions $z \oplus b^j$ can be in error. To see why, observe that, by choice of the subspaces $\B^j$, the last $(m-r-\ell)$ coordinates of $b^j$ are all $0$s. Hence, by Lemma~\ref{lem:projvector}(i), we have $(z \oplus b^j) \langle m-r-\ell \rangle = z \langle m-r-\ell \rangle =: \gamma_z$. From Property~1, $\pi^{(i-1)}(\gamma_z) \leq \beta$. Hence, at most $\beta$ of the positions $z\oplus b^j$ are in error, which proves Property~2.


Property~2 shows, in particular, that the projection at level 0, i.e., $y$, is decoded correctly, which proves the theorem.
\end{proof}

The complexity of the decoder described in Theorem~\ref{thm:main} is $O\big((2\beta+1)^r \cdot n\big)$. Note that a special case of the error patterns of interest is burst errors. If there is a burst of length up to $2^{m-r-\ell}$, then the error positions can span across at most $2^{m-r-\ell}$ consecutive bits, therefore the last $m-r-\ell$ coordinates of all error positions will be distinct. In other words, $\pi(\gamma) \leq 1$ for all $\gamma \in \F_2^{m-r-\ell}$. Thus, a single burst can contribute at most one to each $\pi(\gamma)$ and, consequently, for $\beta$ bursts, $\pi(\gamma) \leq \beta \ \forall \gamma$. Thus, we get the following corollary:

\begin{corollary} \label{cor:r=0}
  The RPA decoder described in Theorem~\ref{thm:main} can correct $\beta \le 2^{\ell}-1$ bursts, each of length up to $2^{m-r-\ell}$, using $2\beta+1$ subspaces at each level of the recursion, provided that the total number of errors is less than $2^{m-r-1}$. The complexity of this decoder is $O\big((2\beta+1)^r \cdot n\big)$.
\end{corollary}

Setting $\ell=1$ above, we see that we can correctly decode a single burst of length up to $2^{m-r-1}-1$ using only three one-dimensional subspaces at each level of recursion.

Let us take an example. Consider the code $\mathrm{RM}(5,2)$, with blocklength $2^m=32$ and minimum distance $d_{\mathrm{min}}=2^{m-r}=8$. Hence, any error pattern of weight less than $d_{\mathrm{min}}/2=4$ can be corrected. Consider a burst error pattern with $E=\{01111,10000,10001\}$ of length $3$. The last two coordinates of all error positions are different. We show that these errors can be corrected by the RPA decoder using three one-dimensional subspaces. We select the following three one-dimensional subspaces for the first level of projection: $\B^1=\{00000,00100\}$, $\B^2=\{00000,01100\}$, and $\B^3=\{00000,10000\}$. For each erroneous position, Table~\ref{tab:error_coset} shows the other positions paired with it under the corresponding one-dimensional subspace. 

\vspace{-2mm}
\begin{table}[h]
\centering
\caption{Positions paired with the three erroneous positions under the three selected subspaces.}
\label{tab:error_coset}
\begin{tabular}{c|ccc}
\textbf{Error position} & $\B^1$ & $\B^2$ & $\B^3$ \\ \hline
$\mathbf{01111}$ & $01011$ & $00011$ & $11111$ \\
$\mathbf{10000}$ & $10100$ & $11100$ & $00000$ \\
$\mathbf{10001}$ & $10101$ & $11101$ & $00001$ 
\end{tabular}
\end{table}

The three selected subspaces have nonzero vectors whose last two coordinates are zero. Hence, each coset generated by these subspaces contains positions having the same last two coordinates. Since the three error positions have distinct last two coordinates, the positions paired with different error positions are disjoint. This can be verified from Table~\ref{tab:error_coset}. Thus, each position is paired or forms coset with at most one erroneous position across the three selected subspaces. Therefore, in the aggregation step, using majority voting, all bits are decoded correctly provided that all projections are decoded correctly.

Now, the errors in $y$ at positions $01111,10000,10001$ appear in the projection $y_{/\B^1}$ at positions $0111,1000,1001$, in $y_{/\B^2}$ at positions $0011,1000,1001$, and in $y_{/\B^3}$ at positions $1111,0000,0001$, respectively. 
Thus, in each of the three projections, the last two coordinates of all error positions are different. Moreover, for each projection, the number of errors is $3<d_{\mathrm{min}}/2$, where the projected code is $\mathrm{RM}(4,1)$. We do not show the complete decoding process here; each projection can be decoded correctly by further projecting it to $\mathrm{RM}(3,0)$. Hence, all three projections can be decoded correctly, and the correctly decoded projections can then be used in the aggregation step to recover the transmitted codeword.  \\

The version of the RPA decoder we have used so far uses $r=0$ as the base case. For constant $\beta$, it has lower complexity than the version of the RPA decoder that uses $r=1$ as the base case. But as seen above, if we use $r=1$ as the base case and use FHT decoder to decode the quotient code $\mathrm{RM}(m-r+1,1)$, then we can choose a larger value of $\beta$, i.e., $\beta \leq 2^{\ell+1}-1$, which allows our framework to cover more error patterns. 

\begin{theorem} \label{thm:r=1}
  For a code $\mathrm{RM}(m,r)$, suppose the number of errors is $t < \frac{d_{\mathrm{min}}}{2}$. Let $\pi(\gamma)= \#\{e \in E: e \langle m-r-\ell \rangle= \gamma \}$ for some fixed $\ell \in \{1,2,\ldots,m-r\}$ and $\gamma \in \F_2^{m-r-\ell}$. If $\pi(\gamma) \leq \beta$ holds for all $\gamma\in\F_2^{m-r-\ell}$, where $\beta \le 2^{\ell+1}-1$ is some fixed integer, then these errors can be corrected by the RPA decoder using $2\beta+1$ one-dimensional subspaces at each level of the recursion, with $r=1$ as the base case.
\end{theorem}

We do not provide the complete proof of Theorem~\ref{thm:r=1}, as it is along the same lines as the proof of Theorem~\ref{thm:main}. The two properties hold in this case as well; the only difference is that, for the proof of Property~2, the base case is the $(r-1)$-th level of projection, corresponding to the quotient code $\mathrm{RM}(m-r+1,1)$. From Property~1, the number of errors at the $(r-1)$-th level is less than $\frac{d_{\mathrm{min}}}{2}=2^{m-r-1}$. Note that $d_{\mathrm{min}}$ is the same for all quotient codes $\mathrm{RM}(m-i,r-i)$ for all $i \in \{1,\ldots,r-1\}$. Thus, the FHT decoder corrects all the errors since it is a maximum-likelihood decoder.

The complexity of the decoder described in Theorem~\ref{thm:r=1} is $O \big( (2\beta+1)^{r-1} \cdot n \log n \big)$. At the expense of a slight increase in complexity, this decoder can handle a larger number of bursts than the one in Theorem~\ref{thm:main}, as can be seen by comparing the corollary below with Corollary~\ref{cor:r=0}.

\begin{corollary} \label{cor:r=1}
  The RPA decoder described in Theorem~\ref{thm:r=1} can correct $\beta \le 2^{\ell+1}-1$ bursts, each of length up to $2^{m-r-\ell}$, using $2\beta+1$ subspaces at each level of the recursion, provided that the total number of errors is less than $2^{m-r-1}$. The complexity of this decoder is $O\big((2\beta+1)^{r-1} \cdot n \log n \big)$.
\end{corollary}

For instance, setting $\ell=1$ above, we see that the decoder of Theorem~\ref{thm:r=1} can correct up to three bursts in which the total number of errors is less than $2^{m-r-1}$ using seven one-dimensional subspaces at each level of the RPA recursion. By way of comparison, the $\ell=1$ instance of Corollary~\ref{cor:r=0} shows that the decoder of Theorem~\ref{thm:main} can correct only one such burst, but it makes use of only three one-dimensional subspaces at each level of the RPA recursion.


\section{Conclusion}
In this paper, we studied the problem of correcting structured error patterns in Reed-Muller codes when the number of errors is less than half the minimum distance. We showed that, by appropriately choosing the subspaces based on the structure of the error pattern, fewer subspaces than those required for arbitrary error patterns can be used for correction. For $\beta$ bursts, we showed that $2\beta+1$ subspaces are sufficient under suitable conditions on the burst lengths. Since $\beta$ can be a constant that is much smaller than the total number of errors $t$, the decoder makes use of a constant number of subspaces at each level of the recursion, even when the number of errors is large. We considered two variants of the RPA decoder, with $r=0$ and $r=1$ as the base case. The $r=0$ base case has lower complexity for constant $\beta$, while the $r=1$ base case allows a larger value of $\beta$, and hence can correct a larger class of structured error patterns.

As future work, we would like to extend this study to error patterns where the number of errors is more than half the minimum distance. In particular, we would like to investigate whether there exist structured error patterns with $t\geq d_{\mathrm{min}}/2$ whose structure can be exploited to correct them using fewer subspaces.

\section*{Acknowledgement} This work was supported in part by a Qualcomm Innovation Fellowship awarded to the first author.

\section*{Statement of AI Usage}
During the preparation of this manuscript, the authors used ChatGPT (OpenAI) to improve the language and readability of the text. The core technical content, analysis, and conclusions were entirely generated and reviewed by the human authors, who remain fully responsible for the accuracy and integrity of the work.

\bibliographystyle{ieeetr}
\bibliography{references}

\end{document}